\documentclass{iopart}
\usepackage[latin1]{inputenc}
\usepackage{iopams} 
\usepackage{color}

\newtheorem{theorem}{Theorem}
\newtheorem{lemma}{Lemma}
\newtheorem{definition}{Definition}

\newenvironment{proof}[1][Proof]
    {\begin{trivlist}\item[\hskip \labelsep \textit{#1.}]}
    {$\square$ \end{trivlist}}

\newcommand{\R}{\mathbb{R}}

\renewcommand{\H}{\mathcal{H}}

\begin{document}
\title{Domains of time-dependent density-potential mappings}
\author{M Penz$^1$ and M Ruggenthaler$^2$}
\address{$^1$ Institute of Theoretical Physics, University of Innsbruck, Austria}
\address{$^2$ Nanoscience Center, University of Jyv\"askyl\"a, Finland}
\ead{markus.penz@uibk.ac.at}

\begin{abstract}
The key element in time-dependent density functional theory is the one-to-one correspondence between the one-particle density and the external potential. In most approaches this mapping is transformed into a certain type of Sturm-Liouville problem. Here we give conditions for existence and uniqueness of solutions and construct the weighted Sobolev space they lie in. As a result the class of $v$-representable densities is considerably widened with respect to previous work.
\end{abstract}

\pacs{31.15.ee, 02.30.Sa}
\submitto{\JPA}

\section{Introduction}

Imagine an interacting many-particle quantum system on a bounded space domain $\Omega \subset \R^3$ within a time interval $[0,T]$ governed by Schrödinger evolution from a fixed initial state $\Psi_0 \in L^2(\Omega^{N})$, where $N$ is the number of particles. Using straightforward computational techniques one soon runs into complexity issues because of the multidimensional structure of the problem. One strategy to bypass this is given by time-dependent density functional theory (TDDFT). The key ingredient is to switch from the wave-function $\Psi$ as a fundamental functional variable for evaluating observables to the one-particle density $n : [0,T] \times \Omega \rightarrow \R_{\geq 0}$. That there exists indeed a one-to-one relation $\Psi \leftrightarrow n$ was first shown by Runge and Gross \cite{runge-gross} for the time-dependent case but without considerations towards the domains of the involved mappings.\\

In order to calculate the one-particle density $n$ without the full complexity of the interacting Schrödinger problem one can rely on two possible schemes. Firstly, quantum fluid dynamical approximations like Thomas-Fermi theory and secondly, the usually adopted Kohn-Sham scheme. \cite{kohn-sham} Therein one substitutes the interacting quantum system by a non-interacting system in an external effective scalar potential $v : [0,T] \times \Omega \rightarrow \R$ leading to exactly the same one-particle density. This amounts to the fundamental question of (non-interacting) $v$-representability. In order to guarantee the existence of such an effective potential one uses the divergence of the local force equation \cite{van-leeuwen}
\begin{equation}
\label{sl-nonlinear}
 -\nabla (n \nabla v) = q[v] - \partial_t^2 n.
\end{equation}

Here the term $q[v] : [0,T] \times \Omega \rightarrow \R$ gives the divergence of the internal local forces, i.e. kinetic terms and all interactions. It is defined as an expectation value $\langle \Psi[v], \hat{q}\, \Psi[v] \rangle$ and thereby involves the state $\Psi[v]$ evolved under the influence of the external potential $v$. The dependence of $q[v]$ on $v$ is therefore highly non-linear.\\

To solve equation (\ref{sl-nonlinear}) one will linearize its right hand side, e.g. by assuming analyticity in time of $n$ and $v$ hence finding a set of coupled linear partial differential equations. \cite{van-leeuwen} A different approach currently under consideration by the authors tries to eliminate the analyticity restriction by pursuing a fixed point scheme to solve (\ref{sl-nonlinear}). \cite{tddft2} In any case one arrives at linear partial differential equations of Sturm-Liouville type for any time $t \in [0,T]$
\begin{equation}
\label{sl-general}
 -\nabla (n \nabla v) = \zeta.
\end{equation}

The time-dependent generalization of density functional theory is nowadays a widely used technique in different fields of physics. \cite{tddft-book} But in contrast to the ground-state theory there have not been a lot of rigorous mathematical investigations of TDDFT. As pointed out above the linear Sturm-Liouville equation is of fundamental importance to the foundations of the theory. Therefore it is crucial to have rigorous results concerning the properties of (\ref{sl-general}) in order to build a sound mathematical  basis for TDDFT. So far uniqueness and existence of solutions $v$ to (\ref{sl-general}) was shown in \cite{tddft1} but under the condition that the density is not only bounded but gapped away from zero for the whole domain $\Omega$. In this work the authors considerably widen the class of allowed densities and give exact domains for the involved external potentials $v$ as well as the inhomogeneity $\zeta$. Further, we discuss implications on the involved external potentials and the one-particle densities within the different approaches \cite{van-leeuwen, tddft2} to solve the non-linear equation (\ref{sl-nonlinear}).

\section{A problem-adapted weighted Sobolev space of potentials}

In the search for solutions to (\ref{sl-general}) we consider weak solutions, which are defined by adjoining an arbitrary $u$ by means of the standard $L^2$ scalar product in all space-variables.
\begin{equation}
-\langle u, \nabla (n \nabla v) \rangle = \langle u,\zeta \rangle
\end{equation}

If we consider only potentials vanishing at the border of $\Omega$ partial integration defines a bilinear form $Q$ by
\begin{equation}
\label{Q-def}
Q(u,v) = \langle \nabla u, n \nabla v \rangle = \langle u,\zeta \rangle.
\end{equation}

The questions of existence and uniqueness of a solution $v$ to (\ref{Q-def}) can now be answered by the theorem of Lax-Milgram. \cite{blanchard-bruening}

\begin{theorem}
\label{lax-milgram} Let $Q$ be a coercive continuous bilinear form on
a Hilbert space $\H$. Then for every continuous linear
functional $\zeta$ on $\H$, there exists a unique $v \in
\H$ such that $Q(u,v) = \zeta(u)$ holds for all $u \in \H$. 
\end{theorem}

A bilinear form $Q$ is said to be coercive if there exists a
constant $c > 0$ such that $Q(u,u) \geq c \|u\|_\H^2$ for all $u \in
\H$. Continuity means we find a $C >0$ such that $Q(u,v) \leq C \|u\|_\H \cdot \|v\|_\H$ for all $u,v \in \H$. $\H$ as a Hilbert space should now be chosen in a way that makes $Q$ coercive as well as continuous. The use of the $L^2$ scalar product from before suggests $\H \subset L^2(\Omega)$. Observing that $\langle\cdot,\cdot\rangle$ and $Q$ can be naturally combined to a bilinear form
\begin{equation}
\label{weighted-sp}
\langle u,v \rangle_\H = \langle u,v \rangle + Q(u,v) = \langle u,v \rangle + \langle \nabla u, n \nabla v \rangle
\end{equation}
we can ask if this yields an adequate scalar product to define our Hilbert space $\H$, i.e. we have to check if
\begin{equation}
\|u\|_\H = \sqrt{\langle u,u \rangle_\H} = \sqrt{\langle u,u \rangle + \langle \nabla u, n \nabla u \rangle}
\end{equation}
is actually a norm. One obvious restriction is $n \geq 0$ but this is certainly true for a one-particle density. The second restriction is such to make sense of the notation ``$\nabla u$'' if only in a distributional sense. Because of $\Omega$ bounded it holds $L^2(\Omega) \subset L^1(\Omega)$ and as the norm of $\H$ is constructed in such a way that $\|u\|_2 \leq \|u\|_\H$ we have the supposed property $\H \subset L^2(\Omega)$. Remembering that elements of $L^1_{\mathrm{loc}}$ can naturally be identified with distributions we are led to the following chain of inclusions allowing elements of $\H$ to be distributionally differentiated.
\begin{equation}
\H \subset L^2(\Omega) \subset L^1(\Omega) \subset L^1_{\mathrm{loc}}(\Omega) \subset \mathcal{D}'(\Omega)
\end{equation}

If $n=1$ then $\|\cdot\|_\H$ is just the norm of the Sobolev space $W^{1,2}(\Omega) = H^1(\Omega)$ and written as $\|\cdot\|_{1,2}$ so we adopt the notation $H^1(\Omega,n)$ for the complete normed space equipped with $\|\cdot\|_\H = \|\cdot\|_{1,2,n}$ and call it a ``weighted Sobolev space''. \cite{adams, kufner-opic}\\

Let us get our notation of the different norms involved straight by defining them all for general $p \in [1,\infty)$. Note the use of the weighting function $n$ only in the second term of the definition of $\|\cdot\|_{1,p,n}$.
\begin{eqnarray}
\|u\|_p &=& \left( \int_\Omega |u|^p \rmd x \right)^{\frac{1}{p}} \\
\|u\|_{1,p} &=& \left( \|u\|_p^p + \left\|\nabla u\right\|_p^p \right)^{\frac{1}{p}} \\
\|u\|_{p,n} &=& \left\|u \, n^{\frac{1}{p}}\right\|_p = \left( \int_\Omega |u|^p \,n \rmd x \right)^{\frac{1}{p}} \\
\|u\|_{1,p,n} &=& \left( \|u\|_p^p + \left\|\nabla u\right\|_{p,n}^p \right)^{\frac{1}{p}}
\end{eqnarray}

We still must not forget the restriction to functions which vanish at the border of $\Omega$ in order to justify the integration by parts used to derive (\ref{Q-def}). This is of course not generally true for elements of $H^1(\Omega,n)$ but can be met if one adopts the usual definition of $H_0^1(\Omega)$ to weighted Sobolev spaces. In that we take the space of infinitely differentiable functions on $\Omega$ with compact support $\mathcal{C}^\infty_0(\Omega)$ and form the closure under our weighted Sobolev norm $\|\cdot\|_{1,2,n}$. The resulting space $\H = H_0^1(\Omega,n)$ equipped with scalar product (\ref{weighted-sp}) is complete and thus a full-fledged Hilbert space of functions which vanish at the border of $\Omega$. This will be the main space of our further investigations.\\

Let us also define the dual of this Hilbert space $H^{-1}(\Omega,n) = (H_0^1(\Omega,n))'$, i.e. the space of linear continuous functionals on $H_0^1(\Omega,n)$. In this we follow standard notation, cf. \cite{adams} 3.12 and 3.13. $\zeta$ from Theorem \ref{lax-milgram} is thought of being an element of this space but we will rather concentrate on regular distributions as we try to find weak solutions to (\ref{sl-general}).

\section{Embedding theorems}

To prove coercivity of $Q$ we follow a strategy largely outlined in \cite{drabek}. The idea is to continuously embed the weighted Sobolev space into a non-weighted one and further into an ordinary $L^p$ space.

\begin{definition}
Let $V, W$ be Banach spaces with $V \subset W$. We say that $V$ is continuously embedded in $W$ and write $V \hookrightarrow W$, if there is a constant $c \geq 0$ such that for all $v \in V$
\begin{equation}
\|v\|_W \leq c \, \|v\|_V.
\end{equation}
We say that $V$ is compactly embedded in $W$ and write $V \hookrightarrow\hookrightarrow W$, if additionally every bounded sequence in $V$ has a subsequence converging in $W$.
\end{definition}

\begin{lemma}\label{embedding1} Let $\Omega$ be bounded, $p > q \geq 1$ and the weighting function $n$ such that $n^{-s} \in L^1(\Omega)$ for $s = \frac{q}{p-q}$, then
\begin{equation}
W_0^{1,p}(\Omega, n) \hookrightarrow W_0^{1,q}(\Omega).
\end{equation}
\end{lemma}

\begin{proof}
Using Hölder's inequality with $\frac{1}{p'} + \frac{1}{q'} = \frac{q}{p} + \frac{p-q}{p} = 1$ we derive
\begin{eqnarray}
\left\| \nabla u \right\|_q &=& \left\| \, |\nabla u|^q \right\|_1^{\frac{1}{q}} = \left\| \left( |\nabla u|^q n^{\frac{q}{p}} \right) n^{-\frac{q}{p}} \right\|_1^{\frac{1}{q}} \\
&\leq& \left( \left\| |\nabla u|^q n^{\frac{q}{p}} \right\|_{\frac{p}{q}} \left\| n^{-\frac{q}{p}} \right\|_{\frac{p}{p-q}}\right)^{\frac{1}{q}} = \left\| |\nabla u|^p n \right\|_1^{\frac{1}{p}} \left\| n^{-s} \right\|_1^{\frac{p-q}{pq}}
\end{eqnarray}
and thus
\begin{equation}\label{inequ-q-p,n}
\| \nabla u \|_q \leq c \| \nabla u \|_{p,n}.
\end{equation}

Now we can easily establish the inclusion, considering that $L^p(\Omega) \subset L^q(\Omega)$ for $\Omega$ bounded.
\end{proof}

The following Lemma is part of the Rellich-Kondrachov Theorem. \cite{adams}

\begin{lemma}\label{embedding2}
Let $\Omega$ be bounded then we have a compact embedding
\begin{equation}
W_0^{m,q}(\Omega) \hookrightarrow\hookrightarrow L^r(\Omega)
\end{equation}
for $1 \geq \frac{1}{r} > \frac{1}{q} - \frac{m}{d}$, provided $m \geq 1$ and $m q < d$.
\end{lemma}

\begin{theorem}\label{embedding3}
Let $\Omega$ be bounded with dimension $\geq 2$ and the weighting function $n$ such that $n^{-2} \in L^1(\Omega)$ then we have a compact embedding
\begin{equation}
H_0^1(\Omega, n) \hookrightarrow\hookrightarrow L^2(\Omega).
\end{equation}
\end{theorem}

\begin{proof}
In the case of Lemma \ref{embedding1} choose $p = 2, q = \frac{4}{3}$ which leads to $s = 2$. Now we can choose $r=2$ and use Lemma \ref{embedding2} to get the following sequence of embeddings.
\begin{equation}
H_0^1(\Omega, n) = W_0^{1,2}(\Omega, n) \hookrightarrow W_0^{1,\frac{4}{3}}(\Omega) \hookrightarrow\hookrightarrow L^2(\Omega)
\end{equation}
\end{proof}

\section{Weak solutions to the Sturm-Liouville problem}

\begin{theorem}\label{sl-solutions}
Let $\Omega$ be bounded with dimension $\geq 2$ and the weighting function $n$ such that $n^{-2} \in L^1(\Omega)$. The bilinear form $Q$ defined on the Hilbert space $H_0^1(\Omega,n)$ by (\ref{Q-def}) is continuous as well as coercive and therefore admits the use of Theorem \ref{lax-milgram}.
\end{theorem}

\begin{proof} \textit{Part 1: Continuity.}
\begin{eqnarray}
|Q(u,v)| &=& \left|\langle \nabla u,n \nabla v \rangle\right| 
= \left|\left\langle \sqrt{n}\, \nabla u,\sqrt{n}\, \nabla v \right\rangle\right| \\
&\leq& \left\| \sqrt{n}\, \nabla u \right\|_2 \cdot \left\| \sqrt{n}\, \nabla v \right\|_2 
= \left\| \nabla u \right\|_{2,n} \cdot \left\| \nabla v \right\|_{2,n} \\
&\leq& \left\| u \right\|_{1,2,n} \cdot \left\| v \right\|_{1,2,n}
\end{eqnarray}

\textit{Part 2: Coercivity.} We need to show $Q(u,u) \geq c \| u \|^2_{1,2,n}$ for all $u \in H_0^1(\Omega,n)$. We start by observing $Q(u,u) = \| \nabla u \|_{2,n}^2$. First we use the result (\ref{inequ-q-p,n}) in the proof of Lemma \ref{embedding1} for the case $p=2, q=\frac{4}{3}$.
\begin{equation}\label{coercivity-res1}
\| \nabla u \|_{\frac{4}{3}} \leq c_1 \| \nabla u \|_{2,n}
\end{equation}
As by Lemma \ref{embedding1} all elements of $H_0^1(\Omega,n)$ are also in $W_0^{1,q}(\Omega)$ we can apply the Poincaré inequality (cf. \cite{adams} 6.30) $\|u\|_q \leq c \|\nabla u\|_q$ to the l.h.s. It immediately follows
\begin{equation}\label{coercivity-res2}
\| u \|_{\frac{4}{3}} \leq c_2 \| \nabla u \|_{2,n}.
\end{equation}
If we combine results (\ref{coercivity-res1}) and (\ref{coercivity-res2}) after taking $(\cdot)^{\frac{4}{3}}$ we get
\begin{equation}
\| u \|_{\frac{4}{3}}^{\frac{4}{3}} +  \| \nabla u \|_{\frac{4}{3}}^{\frac{4}{3}} \leq \left( c_1^{\frac{4}{3}} + c_2^{\frac{4}{3}} \right) \| \nabla u \|_{2,n}^{\frac{4}{3}}
\end{equation}
which leads us straight to the norm of $W^{1,\frac{4}{3}}(\Omega)$ and
\begin{equation}
\|u\|_{1,\frac{4}{3}} \leq c_3 \| \nabla u \|_{2,n}.
\end{equation}
Now we make use of Lemma \ref{embedding2}, which tells us in the case $m = 1, q = \frac{4}{3}, r = 2$
\begin{equation}
\| u \|_2 \leq c_4 \| u \|_{1,\frac{4}{3}}
\end{equation}
and therefore the Hardy inequality \cite{opic-kufner-hardy}
\begin{equation}\label{coercivity-res3}
\| u \|_2 \leq c_5 \| \nabla u \|_{2,n}
\end{equation}
holds. It is now easy to arrive at the desired inequality by squaring (\ref{coercivity-res3}) and adding another $\| \nabla u \|_{2,n}^2$.
\begin{equation}
\| u \|_{1,2,n}^2 \leq \left( c_5^2 + 1 \right) \| \nabla u \|_{2,n}^2
\end{equation}
\end{proof}

Closely related we have the following Theorem as a solution to the general eigenvalue problem for $Q$.

\begin{theorem}
Given a continuous and coercive bilinear form $Q$ on $H_0^1(\Omega,n)$ under the conditions of Theorem \ref{sl-solutions} there is a monotone increasing sequence $(\lambda_m)_{m \in \mathbb{N}}$ of eigenvalues,
\begin{equation}
	0 < \lambda_1 \leq \lambda_2 \leq \lambda_m \stackrel{m \rightarrow \infty}{\longrightarrow} \infty,
\end{equation}
and an orthonormal basis $\{e_m\}_{m \in \mathbb{N}} \subset H_0^1(\Omega,n)$ of $L^2(\Omega)$ such that for all $u \in H_0^1(\Omega,n)$ and all $m \in \mathbb{N}$
\begin{equation}
	Q(u,e_m) = \lambda_m \langle u,e_m \rangle.
\end{equation}
\end{theorem}

\begin{proof}
By Theorem \ref{embedding3} we have a compact embedding $H_0^1(\Omega, n) \hookrightarrow\hookrightarrow L^2(\Omega)$ and this makes Theorem 6.3.4 in \cite{blanchard-bruening} applicable which yields just the given proposition.
\end{proof}

\section{Conclusions and Implications}

We conclude that it is possible to find a unique solution $v_t \in H_0^1(\Omega,n_t)$ to (\ref{sl-general}) under the restrictions given in Theorem \ref{sl-solutions} and $\zeta_t \in H^{-1}(\Omega,n_t)$ for all $t \in [0,T]$. Note that we made the time dependence of the involved sizes and thereby of the whole Hilbert space explicit by adding an index $t$. Even the space-domain $\Omega$ can be made time-dependent in this fashion. Further we have shown the existence of an eigenbasis of $L^2(\Omega)$ of the bilinear form $Q$ from above with eigenvalues $\lambda_m > 0$. \\

With respect to previous considerations \cite{tddft1}, where the linearization scheme of (\ref{sl-nonlinear}) poses restrictions only on the initial density $n_0$, we have extended the set of allowed densities. Especially we have lifted the restriction $\varepsilon \leq n_0 \leq M$ with $M \geq \varepsilon>0$ for $v$-representable densities. However, in \cite{tddft2} one uses a linearization of (\ref{sl-nonlinear}) on the whole time interval. Hence we conclude that 
\begin{equation}
\{ (n_t)_{t \in [0,T]} \,|\, n_t \geq 0, n_t \in L^1(\Omega), n_t^{-2} \in L^1(\Omega)\}
\end{equation}  
is a sufficiently constricted set of densities to guarantee the existence of solutions to the linearized equation (\ref{sl-general}) for the whole time interval. This is an important requirement to make the fixed point approach presented in \cite{tddft2} rigorous and also affects previous considerations in this matter, e.g.~\cite{van-leeuwen}.\\

Finally we want to discuss the condition $\zeta \in H^{-1}(\Omega,n)$. (We will drop the index $t$ again in these last considerations.) Going back to (\ref{sl-nonlinear}) the inhomogeneity $\zeta$ of the linearized equations \cite{van-leeuwen, tddft2} is noted to depend on $q[v] - \partial^2_t n$. So in \cite{van-leeuwen, tddft2} one will naturally assume that $q[v]$ as well as $\partial^2_t n$ are in $H^{-1}(\Omega,n)$. The latter can be seen as an additional condition for $v$-representable densities. By virtue of the continuity equation $\partial_t n = -\nabla \cdot \mathbf{j}$ \cite{tddft-book} we transform this condition, which reads as $\langle u, \partial_t^2 n \rangle < \infty$ for all $u \in H_0^1(\Omega,n)$, into a more elementary form.
\begin{equation}
\langle u, \partial_t^2 n \rangle = -\langle u, \nabla \cdot \partial_t \mathbf{j} \rangle = \langle \nabla u, \partial_t \mathbf{j} \rangle = \left\langle \sqrt{n} \nabla u, n^{-\frac{1}{2}} \partial_t \mathbf{j} \right\rangle
\end{equation}

The components of $\sqrt{n} \nabla u$ are all elements of $L^2(\Omega)$ because of $u \in H_0^1(\Omega, n)$, so the question remains if this is also true for $n^{-\frac{1}{2}} \partial_t \mathbf{j}$. This means we need to have a finite integral involving the force density $\partial_t \mathbf{j}$
\begin{equation}\label{finite-force-integral}
\int \frac{|\partial_t \mathbf{j}|^2}{n}\,\rmd x < \infty.
\end{equation}

Note that this very much resembles the so called Weizsäcker term from Thomas-Fermi theory which converges for finite kinetic energies. \cite{lieb}
\begin{equation}
\int |\nabla \sqrt{n}|^2\,\rmd x = \frac{1}{4} \int \frac{|\nabla n|^2}{n}\,\rmd x < \infty.
\end{equation}

Therefore (\ref{finite-force-integral}) might prove useful in relating our condition on $\partial_t^2 n$ to physical quantities as it is the case for the Weizsäcker term.

\section*{Acknowledgements}
M.R. gratefully acknowledges financial support by the Erwin Schr\"odinger Fellowship J 3016-N16 of the FWF (Austrian Science Fonds).

\end{document}